\documentclass[11pt]{article}

\usepackage{latexsym}
\usepackage{amsmath, amssymb} %
\usepackage{xspace}
\usepackage{multirow}

\newtheorem{xdefinition}{Definition}
\newtheorem{xobservation}{Observation}
\newtheorem{xtheorem}{Theorem}
\newtheorem{xlemma}{Lemma}
\newtheorem{xproposition}{Proposition}
\newtheorem{xcorollary}{Corollary}
\newenvironment{definition}{\begin{xdefinition}\rm}%
{\hspace*{\fill}\raisebox{-1pt}{\boldmath$\Box$}\end{xdefinition}}
{\hspace*{\fill}\raisebox{-1pt}{\boldmath$\Box$}\end{xobservation}}
\newenvironment{theorem}{\begin{xtheorem}\rm}{\end{xtheorem}}
\newenvironment{lemma}{\begin{xlemma}\rm}{\end{xlemma}}
\newenvironment{proposition}{\begin{xproposition}\rm}{\end{xproposition}}

\newenvironment{proof}{\begin{trivlist}\item[]{\bf Proof }}%
{\hspace*{\fill}\raisebox{-1pt}{\boldmath$\Box$}\end{trivlist}}

\newcommand{\opt}{{\ensuremath{\textsc{Opt}}}\xspace}
\newcommand{\maj}{{\ensuremath{\textsc{Maj}}}\xspace}
\newcommand{\nav}{{\ensuremath{\textsc{Nai}}}\xspace}
\newcommand{\eag}{{\ensuremath{\textsc{Eag}}}\xspace}
\newcommand{\majnospace}{{\ensuremath{\textsc{Maj}}}}
\newcommand{\navnospace}{{\ensuremath{\textsc{Nai}}}}
\newcommand{\eagnospace}{{\ensuremath{\textsc{Eag}}}}

\newcommand{\algA}{{\ensuremath{\mathcal{A}}}\xspace}
\newcommand{\algB}{{\ensuremath{\mathcal{B}}}\xspace}

\DeclareMathOperator{\WR}{\textrm{WR}}
\DeclareMathOperator{\Max}{\textrm{Max}}
\DeclareMathOperator{\Min}{\textrm{Min}}

\newcommand{\WEHAVE}{\!:\;}
\newcommand{\TO}{\sqsubseteq}
\newcommand{\TRIP}[3]{\ensuremath{(#1,#2,#3)}\xspace}

\begin{document}

\title{The Frequent Items Problem in \\ Online Streaming
under \\ Various Performance Measures\,\thanks{Supported
in part by the Danish Council for Independent Research.  Part 
of this work was done while the
authors were visitng the University of Waterloo.}}

\author{Joan Boyar \hspace{2em} Kim S. Larsen \hspace{2em} Abyayananda Maiti \\[1ex]
        University of Southern Denmark \\
        Odense, Denmark \\[1ex]
        {\tt \{joan,kslarsen,abyaym\}@imada.sdu.dk}}

\date{}

\maketitle

\begin{abstract}
In this paper, we strengthen the competitive analysis results obtained
for a fundamental online streaming problem, the Frequent Items Problem.
Additionally, we contribute with a more detailed analysis of this
problem, using alternative performance measures,
supplementing the insight gained from competitive analysis.
The results also contribute to the general study of
performance measures for online algorithms. It has long been known
that competitive analysis suffers from drawbacks in certain situations,
and many alternative measures have been proposed. However, more
systematic comparative studies of performance measures have
been initiated recently, and we continue this work, using competitive
analysis, relative interval analysis, and relative worst order
analysis on the Frequent Items Problem.
\end{abstract}

\section{Introduction}
The analysis of problems and algorithms for streaming applications, treating
them as online problems, was started in~\cite{Becchetti09}. 
In online streaming, the
items must be processed one at a time by the algorithm, making some irrevocable
decision with each item. A fixed amount of resources is assumed.
In the frequent items problem~\cite{CH08},
an algorithm must store an item,
or more generally a number of items, in a buffer,
and the objective is to store the items appearing most frequently in
the entire stream.
This problem has been studied in~\cite{Giannakopoulos12}.
In addition to probabilistic considerations,
they analyzed deterministic algorithms using competitive analysis.
We analyze the frequent items problem using relative interval
analysis~\cite{Dorrigiv09} and relative worst order analysis~\cite{Boyar07}.
In addition, we tighten the competitive analysis~\cite{ST85,KMRS88}
results from~\cite{Giannakopoulos12}.

It has been known since the start of the area that competitive analysis does not
always give good results~\cite{ST85} and many alternatives have been proposed.
However, as a general rule, these alternatives have been fairly problem specific
and most have only been compared %
to competitive analysis.
A more comprehensive study of a larger number of performance measures
on the same problem scenarios was initiated in~\cite{BIL09p}
and this line of work has been continued in~\cite{BLM12p,BGL12p,BGLtap}.
With this in mind, we would like to produce complete and tight results,
and for that reason, we focus on a fairly simple combinatorial problem
and on simple algorithms for its solution, incorporating greediness
and adaptability trade-offs to a varying extent.

Finally, we formalize a notion of competitive function, as opposed to
competitive ratio, in a manner which allows us to focus on the constant
in front of the high order term.
These ideas are also used to generalize
relative worst order analysis.

\section{Preliminaries}\label{sec:prob}
This is a streaming problem, but as usual in online algorithms
we use the term sequence or input sequence to refer to
a stream.
We denote an \emph{input sequence} by $I = a_1, a_2,\ldots, a_n$,
where the items $a_i$ are from some universe $\mathcal{U}$,
assumed to be much larger than $n$.
We may refer to the index also as the \emph{time step}.
We consider online algorithms, which means that items are given
one by one. 

We consider the simplest possible frequent items
problem: %
An algorithm has a \emph{buffer} with space for one item.
When processing an item, the algorithm can either discard the item
or replace the item in the buffer by the item being processed.
The objective is to keep the most frequently occurring items in the buffer,
where frequency is measured over the entire input, i.e.,
when an algorithm must make a decision, the quality of the decision
also depends on items not yet revealed to the algorithm.
We define this objective function formally:

Given an online algorithm $\mathcal{A}$ for this problem,
we let $s^{\mathcal{A}}_t$ denote \emph{the item in the buffer at time step~$t$}.
We may omit the superscript when it is clear from the context which
algorithm we discuss.

Given an input sequence $I$ and an item $a\in \mathcal{U}$, the {\em frequency}
of the item is defined as $f_I(a) = \frac{n_I(a)}{n}$,
where $n_I(a) = | \{i\mid a_i=a\}|$ is the number of occurrences of $a$ in $I$.
The objective is to maximize the
\emph{aggregate frequency}~\cite{Giannakopoulos12}, defined by
$F_{\algA}(I)=\sum_{t=1}^n f_I(s^{\mathcal{A}}_t)$,
i.e., the sum of the frequencies of the items stored in the
buffer over the time.

We compare the quality of the achieved aggregate frequencies
of three different deterministic online algorithms
from~\cite{Giannakopoulos12}:
the naive algorithm (\nav), the eager algorithm (\eag),
and the majority algorithm (\maj).
All three are practical streaming 
algorithms, being simple and using very little extra space.

\begin{definition}\label{def:nav}
\mbox{\rm [\navnospace]} \nav buffers every item as it arrives,
i.e., $s^{\nav}_t = a_t$ for all $t=1,2,\ldots, n$.
\end{definition}

The algorithm \eag switches mode upon detecting a {\em repeated item}, an
item which occurs in two consecutive time steps.
\begin{definition}\label{def:eag}
\mbox{\rm [\eagnospace]} Initially, \eag buffers every item as it arrives.
If it finds %
a repeated item, then it keeps
that item until the end, i.e., let
$$t^* =\min_{1\leq t\leq n-1} \{t\mid a_t = a_{t+1}\},$$
if such a $t$ exists, and otherwise $t^* = n$.
Then \eag is the algorithm with $s^{\eag}_t = a_t$ for all $t\leq t^*$ and
$s^{\eag}_t = a_{t^*}$ for all $t>t^*$.
\end{definition}

The algorithm \maj keeps a counter along with the buffer.
Initially, the counter is set to zero.
\begin{definition}\label{def:maj}
\mbox{\rm [\majnospace]}
If the counter is zero, then \maj buffers the arriving item and sets the counter to one.
Otherwise, if the arriving item is the same as the one currently buffered,
\maj increments the counter by one, and otherwise decrements it by one.
\end{definition}

Finally, as usual in online algorithms, we let \opt denote an optimal
offline algorithm. \opt is, among other things, used in competitive
analysis as a reference point, since no online algorithm can do better.
If $\algA$ is an algorithm, we let $\algA(I)$ denote the result (profit)
of the algorithm, i.e., $\algA(I)=F_{\algA}(I)$.

In comparing these three algorithms, we repeatedly
use the same two families of sequences;
one where \eag performs particularly poorly and one where \maj
performs particularly poorly.

\begin{definition}\label{sequences}
We define the sequences
$$E_n = a,a,b,b,\ldots,b,$$ where there are $n-2$ copies of $b$, and
\[ W_n = \left\{ \begin{array}{ll}
a_1,a_0,a_2,a_0,\ldots, a_{\frac{n}{2}},a_0 & \mbox{for even $n$}\\
a_1,a_0,a_2,a_0,\ldots, a_{\lfloor\frac{n}{2}\rfloor},a_0,
a_{\lceil\frac{n}{2}\rceil} & \mbox{for odd $n$}.
\end{array} \right. \]
\end{definition}

The four algorithms, including \opt,
obtain the aggregate frequencies below on these two families of sequences.
The arguments are simple, but fundamental, and also serve as an
introduction to the algorithmic behavior of these algorithms.
\begin{proposition}\label{prop_sequences}
The algorithms' results on $E_n$ and $W_n$ are as in
Fig.~\ref{fig-profit}.
\begin{figure}[ht]
\begin{center}
$\begin{array}{|l||c|c|} \hline
& E_n & W_n \\ \hline\hline
\rule[-3ex]{0em}{7ex}\nav & n-4+\frac{8}{n} & 
\left\{ \begin{array}{ll}
\frac{n}{4} + \frac{1}{2} & \mbox{for even $n$}\\[.5ex]
\frac{n}{4} + \frac{3}{4n} & \mbox{for odd $n$}
\end{array} \right. \\ \hline
\rule[-2ex]{0em}{5ex}\eag & 2 & \mbox{as \nav} \\ \hline
\rule[-2ex]{0em}{5ex}\maj & n-6+\frac{16}{n} & 1 \\ \hline
\rule[-3ex]{0em}{7ex}\opt & \mbox{as \nav} &
\left\{ \begin{array}{ll}
\frac{n}{2} - \frac{1}{2} + \frac{1}{n} & \mbox{for even $n$}\\[.5ex]
\frac{n}{2} -1 + \frac{3}{2n}  & \mbox{for odd $n$}
\end{array} \right. \\ \hline
\end{array}$
\end{center}
\caption{The algorithms' aggregate frequencies on $E_n$ and $W_n$.}
\label{fig-profit}
\end{figure}
\end{proposition}
\begin{proof}
In $E_n$, the frequency of $a$ is $\frac{2}{n}$ and the frequency of $b$ is $\frac{n-2}{n}$.
Thus $\nav(E_n)= 2\frac{2}{n}+ (n-2)\frac{n-2}{n} = n-4+\frac{8}{n}$.
In $W_n$, the frequency of $a_0$ is $\lfloor\frac{n}{2}\rfloor/n$, and the frequencies of all the other $a_i$, $1\leq i \leq \lceil \frac{n}{2}\rceil$, are $\frac{1}{n}$. Thus, $\nav(W_n)= \lceil\frac{n}{2}\rceil\frac{1}{n} + \lfloor \frac{n}{2} \rfloor \frac{\lfloor\frac{n}{2}\rfloor}{n}$. Considering both even and odd $n$ gives the required result.

When processing $E_n$, \eag keeps $a$ in its buffer. Hence, $\eag(E_n)= n\frac{2}{n} = 2$.
Since $W_n$ has no repeated item, $\eag(W_n)=\nav(W_n)$.

For $E_n$, \maj will have $a$ in its buffer for the first four time steps, so $\maj(E_n)$ is $4\frac{2}{n}+(n-4)\frac{n-2}{n} = n-6+\frac{16}{n}$. For $W_n$, \maj brings each $a_i$, $1\leq i\leq n$, into its buffer and never brings $a_0$ into its buffer. Thus, $\maj(W_n) = n\frac{1}{n} = 1$.

With $E_n$, \opt is forced to perform the same as \nav.
In $W_n$, \opt must buffer $a_1$ in the first time step, but it buffers $a_0$ for the remainder of the sequence. Thus, $\opt(W_n) = \frac{1}{n}+(n-1)\frac{\lfloor \frac{n}{2}\rfloor}{n}$. Considering both even and odd $n$ gives the required result.
\end{proof}

\begin{definition}\label{def:worstpermut}
Let ${\cal A}$ be any online algorithm.
We denote the worst aggregate frequency of ${\cal A}$ 
over all the
permutations $\sigma$ of $I$ by
${\cal A}_{W}(I) = \min_{\sigma} {\cal A} (\sigma(I))$.
\end{definition}
It is convenient to be able to consider items in order of their
frequencies.
Let $D(I) = a_1', a_2',\ldots, a_n'$ be a sorted list of the item
in $I$ in nondecreasing order of frequencies.
For example, if $I = a, b, c, a, b, a$, then $D(I) = c,b,b,a,a,a$.
We will use the notation $D(I)$ throughout the paper.
\begin{lemma}\label{lem:wmaj}
For odd $n$,
 $\maj_{W}(I) = 2\sum_{i=1}^{\lfloor \frac{n}{2} \rfloor} f_I(a_i')
+ f_I(a_{\lceil \frac{n}{2}\rceil}')$,
and for even $n$,
$\maj_{W}(I) = 2\sum_{i=1}^{\frac{n}{2}} f_I(a_i')$,
where the $a_i'$ are the items of $D(I)$.
\end{lemma}
\begin{proof}
Every time step where the counter is decremented can be paired with an
earlier one
where it is incremented and the same item is in the buffer.
So, at least $\lceil\frac{n}{2}\rceil$ requests contribute to the aggregate frequency of the algorithm.
One can order the items so that exactly the $\lceil \frac{n}{2}\rceil$
requests to that many least frequent items are buffered as follows:
Assuming $n$ is even, then the worst permutation is $a_1', a_n',  a_2',a_{n-1}',\ldots  a_{\frac{n}{2}}', a_{\frac{n}{2}+1}'$.
All  (but the last request when $n$ is odd) of the requests which lead
to an item entering the buffer contribute twice, since they are also in
the buffer for the next step.
\end{proof}

\section{Competitive Analysis}\label{sec:comp}
An online streaming problem was first studied from an online algorithms
perspective using competitive analysis by Becchetti and 
Koutsoupias~\cite{Becchetti09}.
Competitive analysis\cite{ST85,KMRS88} evaluates an online algorithm in comparison to an optimal offline algorithm.
For a maximization problem, an algorithm, \algA is called $c$-competitive,
for some constant $c$, if there exists a constant $\alpha$ such that for all finite input sequences $I$, $\opt(I) \leq c\cdot \algA(I) +\alpha$.
The competitive ratio of \algA is the infimum over all $c$ such that \algA is $c$-competitive.
Since, for the online frequent items problem, the relative performance of algorithms depends on the length of $I$, we define
a modified and more general version of competitive analysis,
providing a formal basis for our own claims as well as claims made in
earlier related work.
Functions have also been considered in~\cite{DorrigivL05}.
Here, we focus on the constant in front of the most significant term.
Our definition can be adapted easily to minimization problems in the same
way that the adaptations are handled for standard competitive analysis.
In all these definitions, when $n$ is not otherwise defined,
we use it to denote $|I|$, the length of the
sequence $I$. As usual, when using asymptotic notation in inequalities,
notation such as $f(n) \leq g(n) + o(g(n))$ means that there exists a
function $h(n)\in o(g(n))$ such that $f(n) \leq g(n) + h(n)$.
Thus, we focus on the multiplicative factors that relate the online
algorithm's result to the input length.

\begin{definition} \label{def:com}
An algorithm \algA is $f(n)$-{\em competitive} if
\[\forall I\WEHAVE \opt(I) \leq (f(n)+o(f(n)))\cdot\algA(I).\]
\algA has {\em competitive function} $f(n)$ if \algA is
$f(n)$-{\em competitive} and for any $g(n)$ such that \algA
is $g(n)$-{\em competitive},
$\lim_{n \rightarrow  \infty}\frac{f(n)}{g(n)} \leq 1$.

If algorithm \algA has {\em competitive function} $f(n)$ and algorithm \algB has {\em competitive function} $f'(n)$, then \algA is better than \algB according to competitive analysis if
$\lim_{n \rightarrow  \infty}\frac{f(n)}{f'(n)} < 1$.
\end{definition}

Thus, the concept of competitive function is an exact characterization
up to the level of detail we focus on. It can be viewed as an
equivalence relation, and if $\lim_{n\rightarrow\infty}\frac{f(n)}{g(n)}=1$ for two
functions $f(n)$ and $g(n)$, then they belong to (and are representatives of)
the same equivalence class. For example,
$\frac{\sqrt{n}}{2}$ and $\frac{\sqrt{n}}{2-\frac{1}{\sqrt{n}}}$ are
considered equivalent,
whereas $\frac{\sqrt{n}}{2}$ and $\frac{\sqrt{n}}{4}$ are not.

All three algorithms discussed here are non-competitive according to
the original definition.
However, information regarding the relative quality of these algorithms
can be obtained by considering the most significant constants
from the corresponding functions.
Giannakopoulos et~al.\ has proved that no randomized algorithm for the
online frequent items problem, where the buffer has room for one item,
can have a competitive function better than
$\frac{1}{3}\sqrt{n}$~\cite{Giannakopoulos12}.
That result can be strengthened for the deterministic case:
\begin{theorem}\label{thm:comp_all}
No deterministic algorithm for the online frequent items problem can have a competitive function better than $\frac{\sqrt{n}}{2}$.
\end{theorem}
\begin{proof}
Consider any deterministic algorithm $\mathcal{A}$, and input of the form
$$I_n = a_1, a_2, \ldots a_{n-\sqrt{n}}, x, x, \ldots, x$$ where the
first $n-\sqrt{n}$ items are distinct and the last $\sqrt{n}$
items are identical.
Since \algA is deterministic, an adversary will know whether $a_1$
or $a_2$ is in the buffer upon completion of time step~2.
The value of $x$ is based on this.
If it is $a_2$, then the adversary sets $x=a_1$, and if it is $a_1$, then
it sets $x=a_2$.
As $x$ does not occur among the next $n-\sqrt{n}-2$ items,
$\mathcal{A}$ has no chance of bringing $x$ into its buffer until
the last $\sqrt{n}$ items arrive, so it
stores $x$ in its buffer at most $\sqrt{n}+1$ times.
$\opt$ stores $x$ at least $n-1$ times.
That gives the ratio of
\begin{eqnarray}
 \frac{\opt (I_n)}{\mathcal{A}(I_n)} &\geq& \frac{\frac{1}{n} + (n-1)\frac{\sqrt{n} + 1}{n}}{(n-\sqrt{n} - 1)\frac{1}{n} +
(\sqrt{n} + 1)\frac{\sqrt{n} + 1}{n}} \nonumber
\\
&=& \frac{1 + (n-1)(\sqrt{n} +1)}{n- \sqrt{n} -1 + (\sqrt{n} + 1)^2} \nonumber
\\
&=& \frac{n + \sqrt{n} -1}{2\sqrt{n} + 1} \nonumber
\\
&\geq& \frac{\sqrt{n}}{2}, \nonumber \mbox{    for $n\geq 4$}
\end{eqnarray}
\end{proof}

In \cite{Giannakopoulos12}, Giannakopoulos et~al.\ proved that
for all sequences $I$ of length $n$, $\opt(I)\leq \sqrt{n}\cdot\nav(I)$.
Here we give a tighter result for \nav.
\begin{theorem}\label{thm:comp_nav}
\nav has competitive function $\frac{\sqrt{n}}{2}$. It is an 
optimal deterministic online algorithm for the frequent items problem.
\end{theorem}
\begin{proof}
Let $f$ be the frequency of the most frequent item in the input sequence $I$.
Since the lowest possible frequency of an item is $\frac{1}{n}$,
\[
\nav(I) \geq nf^2 + (n-nf)\frac{1}{n}
\mbox{ and }
\opt(I) \leq nf
\]
Thus,
\begin{eqnarray}
\frac{\opt(I)}{\nav(I)} &\leq& \frac{nf}{nf^2 + 1 -f}  \label{eq:comp_nav}
\end{eqnarray}
The right hand side of Ineq.~\ref{eq:comp_nav} reaches its maximum when $f=\frac{1}{\sqrt{n}}$.
Substituting this value into Ineq.~\ref{eq:comp_nav}, we get the result:
\begin{eqnarray}
\frac{\opt(I)}{\nav(I)} &\leq& \frac{\sqrt{n}}{2-1/\sqrt{n}}  = \frac{\sqrt{n}}{2} + \frac{1}{2(2-1/\sqrt{n})} %
\nonumber
\end{eqnarray}
Thus, \nav is a $\frac{\sqrt{n}}{2}$-competitive algorithm and,
by Theorem~\ref{thm:comp_all}, it is optimal.

\end{proof}

For \maj Giannakopoulos et~al.~\cite{Giannakopoulos12} proved a competitive ratio of $\varTheta(n)$.
We give the asymptotically tight bounds, including the multiplicative factor.
\begin{theorem}\label{thm:comp_maj}
\maj has competitive function $\frac{n}{2}$.
\end{theorem}
\begin{proof}
For the lower bound, consider the family of sequences, $W_n$, from
Definition~\ref{sequences}.
By Proposition~\ref{prop_sequences}, $\maj(W_n) = 1$,
and
$$\opt(W_n) = \left\{ \begin{array}{ll}
\frac{n}{2} - \frac{1}{2} + \frac{1}{n} & \mbox{for even $n$}\\[1ex]
\frac{n}{2} -1 + \frac{3}{2n}  & \mbox{for odd $n$}
\end{array} \right.$$
Consequently,
$ \opt(W_n) \geq \frac{n}{2}\maj(W_n) - 1.$
Thus, the competitive function cannot be better than $\frac{n}{2}$.

For the upper bound, let $f$ be the largest frequency of any item
in some input sequence $I$ of length $n$.
\opt cannot have an aggregate frequency larger than $nf$.

If $f \leq \frac{1}{2}$, then, since no algorithm can have an
aggregate frequency
less than one in total, $\frac{\opt(I)}{\maj(I)}\leq nf\leq \frac{n}{2}$.

It remains to consider the range $\frac{1}{2} < f \leq 1$.
Let $a_0$ denote the most frequent item in $I$.
Note that $a_0$ must be in the buffer at some point since $f>\frac{1}{2}$.

Since there are $n-fn$ items different from $a_0$, the total
length of all subsequences where $a_0$ is not in the buffer
is at most $2(n-fn)$.
This means that $a_0$ {\em is} in the buffer at least
$n - 2(n-fn) = 2fn - n$ times, collecting
at least $(2fn - n)f=2nf^2-nf$. The remaining items
collect at least $2(n-fn)\frac{1}{n}$. In total, this
amounts to $2nf^2-nf+2-2f$.
If we can prove that this quantity is at least $2f$ for large
$n$,
then asymptotically,
$\frac{\opt(I)}{\maj(I)}\leq \frac{nf}{2nf^2-nf+2-2f}
 \leq\frac{nf}{2f}=\frac{n}{2}$
and we will be done.
Now,
$ 2nf^2-nf+2-2f \geq 2f $ if and only if
$ 2nf^2-(n+4)f+2\geq 0 $. Taking the derivative of the left
side shows that the left side is an increasing function of $f$
for $n\geq 4$ and $f\geq \frac{1}{2}$. Thus, $\opt(I) \leq \frac{n}{2}\maj(I)$ holds
for all $f$ and all $n\geq 4$. This implies that \maj is $\frac{n}{2}$-competitive and, combined with the lower bound result, that the competitive function of \maj is $\frac{n}{2}$.
\end{proof}

\begin{theorem}\label{thm:comp_eag}
The competitive function of the algorithm \eag is $\frac{n}{2}$.
\end{theorem}
\begin{proof}
For the lower bound, consider the family of sequences, $E_n$, from
Definition~\ref{sequences}.
By Proposition~\ref{prop_sequences}, $\eag(E_n) =2$, and $\opt(E_n) = n-4+\frac{8}{n}$.
Thus, $\opt(E_n) = \frac{n}{2}\eag(E_n) -4 + \frac{8}{n}$, and \eag's competitive function cannot be better than $\frac{n}{2}$.

If there are no %
repeated items in $I$, then \eag behaves like \nav and that will give $\opt(I) \leq (\frac{\sqrt{n}}{2}+o(\sqrt{n}))\eag(I)$ by Theorem~\ref{thm:comp_nav}.
It is evident from the lower bound result that the competitive function for \eag  is worse than
$\frac{\sqrt{n}}{2}$, so we assume that there is at least one repeated item in $I$.
Let time steps $p+1$ and $p+2$ be the first  occurrence of a repeated
item in $I$.
Let $b$ be the most frequent item in $I$. Note that $b$ is not necessarily the
item which arrived at time steps $p+1$ and $p+2$.
After $p$, all the items could conceivably be $b$, but among the first $p$ items, at most $\frac{p}{2}$ items can be $b$, because $p+1$ and $p+2$ are the indices of the first repeated item.
So, an upper bound on the maximum frequency, $f_I(b)$, is $\frac{n-p + \frac{p}{2}}{n} = \frac{n-\frac{p}{2}}{n}$.
This gives an upper bound of $\opt(I) \leq n\frac{n-\frac{p}{2}}{n} = n-\frac{p}{2}$.

Now we  consider a lower bound on $\eag(I)$.
In the worst case for \eag, all the items before $p+1$ are
distinct, so their contribution to $\eag(I)$ is at least $\frac{p}{n}$.
In the worst case for \eag, the item that occurs at
time steps $p+1$ and $p+2$ has frequency $\frac{2}{n}$, so the contribution
to $\eag(I)$ from the items after $p$ is at least  $(n-p)\frac{2}{n}$.
Thus, $\eag(I) \geq \frac{p}{n} + (n-p)\frac{2}{n} = 2 - \frac{p}{n}$,
and $$\frac{\opt(I)}{\eag(I)} \leq \frac{n-\frac{p}{2}}{2-\frac{p}{n}}=\frac{n}{2}.$$
Hence, \eag has competitive function $\frac{n}{2}$.
\end{proof}

\section{Relative Interval Analysis}\label{sec:realint}
Dorrigiv et~al.~\cite{Dorrigiv09} proposed another analysis method, relative interval analysis, in the context of paging.
Relative interval analysis compares two online algorithms directly, i.e., it does not use the optimal offline algorithm as the baseline of the comparison.
It compares two algorithms on the basis of the rate of the outcomes over the length of the input sequence rather than their worst case behavior.
Here we define this analysis for maximization problems for two algorithms $\mathcal{A}$ and $\mathcal{B}$, following~\cite{Dorrigiv09}.
\begin{definition}\label{def:relint}
Define
\[\Min_{\mathcal{A}, \mathcal{B}}(n) = \min_{|I| = n} \left \lbrace \mathcal{A}(I) - \mathcal{B}(I)\right\rbrace
\mbox{~and~}
\Max_{\mathcal{A}, \mathcal{B}}(n) = \max_{|I| = n} \left \lbrace \mathcal{A}(I) - \mathcal{B}(I)\right\rbrace,\]
and 
\[\Min(\mathcal{A}, \mathcal{B}) = \liminf _{n\to \infty} \frac{\Min_{\mathcal{A}, \mathcal{B}}(n)}{n} \mbox{ and } \Max(\mathcal{A}, \mathcal{B}) = \limsup _{n\to \infty} \frac{\Max_{\mathcal{A}, \mathcal{B}}(n)}{n}.\]
The {\em relative interval} of $\mathcal{A}$ and $\mathcal{B}$ is defined as
$$l(\mathcal{A}, \mathcal{B}) = \left[ \Min(\mathcal{A}, \mathcal{B}), \Max(\mathcal{A}, \mathcal{B}) \right].$$
If $\Max(\mathcal{A}, \mathcal{B}) > |\Min(\mathcal{A}, \mathcal{B})| $, then $\mathcal{A}$ is said to have {\em better performance} than $\mathcal{B}$ in this model.
\end{definition}
Note that $\Min(\mathcal{A}, \mathcal{B}) = - \Max(\mathcal{B}, \mathcal{A})$ and $\Max(\mathcal{A}, \mathcal{B}) = -\Min(\mathcal{B}, \mathcal{A})$.

For any pair of algorithms, \algA and \algB, for the frequent items
problem, there is a trivial upper bound on $\Max(\algA,\algB)$ and
lower bound on $\Min(\algA,\algB)$.
\begin{proposition}\label{trivialMax}
For any pair of algorithms \algA~ and \algB, $\Max(\algA,\algB)\leq 1$
and $\Min(\algA,\algB)\geq -1$.
\end{proposition}
\begin{proof}
The maximum aggregate frequency any algorithm could
have is for a sequence where all items are identical, giving the value $n$.
The minimum is for a sequence where all items are different, giving the
value $1$. The required bounds follow since $\limsup_{n\rightarrow\infty} \frac{n-1}{n} = 1$.
\end{proof}

\subsection{Naive vs.\ Eager}
According to relative interval analysis, \nav has better performance than \eag.
\begin{theorem}\label{thm:relint_nav_eag}
 According to relative interval analysis $l(\nav, \eag) = [-\frac{1}{4},1]$.
\end{theorem}
\begin{proof}
By Proposition~\ref{trivialMax}, $\Max(\nav, \eag)\leq 1$.

We now consider a lower bound on $\Max(\nav, \eag)$.
By Proposition~\ref{prop_sequences},
we have that
$ \nav (E_n) - \eag (E_n)  = (n-4+\frac{8}{n})-2$, so
\[
\limsup _{n\to \infty} \frac{\nav (E_n) - \eag (E_n)}{n}
= \limsup _{n\to \infty} \frac{n-6+\frac{8}{n}}{n} = 1.
\]
Thus, $\Max(\nav, \eag)= 1$.

We now consider $\Min(\nav, \eag)$.
For the upper bound on the minimum value of $\nav (I) - \eag (I)$,
let $I$ contain $r=\lceil \frac{n+1}{2} \rceil$ copies of $a$ and $\lfloor
\frac{n-1}{2}\rfloor$ distinct items
 $a_1, a_2, \ldots a_{n-r}$, and let $I$ start with $a a$.
For this sequence, \nav's aggregate frequency is
$\lceil \frac{n+1}{2}\rceil \frac{\lceil\frac{n+1}{2}\rceil }{n} + \lfloor
\frac{n-1}{2}\rfloor \frac{1}{n}$,
which is  $\frac{n}{4} + \frac{3}{2}$ if $n$ is even and
$\frac{n}{4} +1 -\frac{1}{4n}$ if $n$ is odd.
\eag's aggregate frequency is
$n \frac{\lceil\frac{n+1}{2} \rceil}{n}$, which is
$ \frac{n}{2} + 1$ if $n$ is even and $\frac{n+1}{2}$ if $n$ is odd.
This gives an upper bound of $\Min_{\nav,\eag}(n)\leq
\nav (I) - \eag (I)$, which is $-\frac{n}{4} +\frac{1}{2}$ if $n$
is even, and $-\frac{n}{4} +\frac{1}{2}-\frac{1}{4n}$ if $n$ is odd.
Thus,
\begin{equation}
\Min(\nav, \eag)\leq \liminf _{n\to \infty} \frac{\frac{1}{2} - \frac{n}{4}}{n} = -\frac{1}{4}. \nonumber
\end{equation}

Next we calculate a lower bound on $\Min(\nav, \eag)$.
Assume that among sequences of length $n$, $I$ gives the smallest
possible value of $\nav(I)-\eag(I)$.
From the definitions of \nav and \eag, it is evident that there must
be a repeated item if $\nav (I) - \eag (I) <0$.
Suppose the first repeated item is item $a$ at time steps $p+1$ and $p+2$.
Before $p+1$, both \nav and \eag have  the same items in the buffer,
and both \nav and \eag have $a$ in their buffers every time it occurs.
We show that we can assume that all items in $I$ different from
$a$ each occur only once in $I$.

First, suppose that there is an
item $b\not= a$ before $p+1$ with frequency greater than $1$ in $I$.
Replace this occurrence of $b$ by a new item, $b'$, which does not occur
in $I$ to obtain $I'$. The contribution to the aggregate frequency
from $b'$ and any $b$s before $p+1$ is identical for \nav and \eag on $I'$.
The contribution to \eag's aggregate frequency from items after $p$ is
unchanged, but if \nav has any $b$s after $p$, the contribution to \nav's
aggregate frequency from them is lower in $I'$ than in $I$. Thus
$\nav(I')-\eag(I') < \nav(I) - \eag(I)$, contradicting the minimality
for $I$.

Now we can assume that any repeated items other than $a$ occur only after $p$.
Clearly, the same technique of replacing one of these repeated
items by a new item which does not already occur will only affect \nav's
aggregate frequency and only decrease it, contradicting the minimality
of $I$. Thus, we may assume that $a$ is the only repeated item.

We may also assume that the item $a$ does not occur before
time $p+1$, since swapping such an occurrence with the item in
location $p$ has no effect on either \nav's or \eag's aggregate frequency.

Consequently, if the number of occurrences of $a$ is denoted by $n_I(a)$,
then $\nav(I)-\eag(I) = (n-n_I(a))\frac{1}{n} + n_I(a)\frac{n_I(a)}{n} -
(p\frac{1}{n} + (n-p)\frac{n_I(a)}{n})$. Since $n_I(a)>1$, this is
clearly minimized at $p=0$, so the first two occurrences of $a$ are
in the first two locations. Taking the derivative and setting it equal
to zero gives that the minimum occurs when $n_I(a)=\frac{n+1}{2}$. This
gives that $\nav(I)-\eag(I) \geq -\frac{n}{4} +\frac{1}{2}-\frac{1}{4n}$,
and  $\Min(\nav, \eag)=-\frac{1}{4}$.
Thus, $l(\nav, \eag) = [-\frac{1}{4}, 1]$.
\end{proof}

\subsection{Naive vs.\ Majority}

\nav and \maj are equally good according to relative interval analysis.
\begin{theorem}\label{thm:relint_nav_maj}
 According to relative interval analysis $l(\nav, \maj) = [-\frac{1}{4}, \frac{1}{4}]$.
\end{theorem}
\begin{proof}
For the maximum value of $\nav (I) - \maj (I)$, it is sufficient to consider the worst permutation of $I$ for \maj since \nav has the same output for all permutations of $I$.
For the worst permutation, $\maj_W(I)$ will buffer only the first $\lceil \frac{n}{2} \rceil$ items of the distribution $D(I)$.
The first $\lfloor \frac{n}{2} \rfloor$ items will be buffered twice and in case of odd $n$, the $\lceil \frac{n}{2} \rceil$th item will be stored once at the last time step.
Let $D(I) = a_1', a_2', a_3',\ldots, a_n'$. Then
\begin{eqnarray}%
\nav(I) - \maj_W(I)  &=& \sum_{i=1}^n f_I(a'_i) - 2\sum_{i=1}^{\lfloor \frac{n}{2} \rfloor}f_I(a'_i) - \left (\left\lceil \frac{n}{2} \right\rceil - \left\lfloor \frac{n}{2} \right\rfloor \right) f_I(a'_{\lceil \frac{n}{2} \rceil}) \nonumber\\
 &=& \sum_{i=\lceil \frac{n+2}{2}\rceil}^n f_I(a'_i) - \sum_{i=1}^{\lfloor\frac{n}{2}\rfloor}f_I(a'_i). \label{eq:dif_navmaj}
\end{eqnarray}
Let $p$ be the number of occurrences of the most frequent item in $I$.  Then
\begin{eqnarray}
\nav(I) - \maj_W(I) &=&
\sum_{i=\lceil \frac{n+2}{2}\rceil}^n f_I(a'_i) - \sum_{i=1}^{\lfloor\frac{n}{2}\rfloor}f_I(a'_i) \nonumber
\\
&\leq& \left\lfloor\frac{n}{2}\right\rfloor\frac{p}{n} - \left(p - \left\lceil\frac{n}{2}\right\rceil\right)\frac{p}{n} \nonumber
\\
&=& p - \frac{p^2}{n}. \nonumber
\end{eqnarray}
If $n$ is even, an upper bound on the maximum difference will be achieved when $p = \frac{n}{2}$, and for odd $n$ when $p = \frac{n+1}{2}$.
This gives an upper bound on the maximum of $\nav(I) - \maj(I)$ of $\frac{n}{4}$ for even $n$ and $\frac{n}{4} - \frac{1}{4n}$ for odd $n$.
For a lower bound on the maximum value of $\nav (I) - \maj (I)$, we consider the family
of sequences, $W_n$, from Definition~\ref{sequences}. By Proposition~\ref{prop_sequences},
for even $n$,  $\nav (W_n) - \maj (W_n) = \frac{n}{4} - \frac{1}{2}$, and for odd $n$, $\nav (W_n) - \maj (W_n) = \frac{n}{4} - 1 + \frac{1}{4n}$.
Thus,
$\Max(\nav, \maj) \geq \limsup _{n\to \infty} \frac{\nav (W_n) - \maj (W_n)}{n} = \frac{1}{4}$,
matching the upper bound.

To derive the minimum value of $\nav (I) - \maj (I)$,
we calculate the maximum value of $\maj (I) - \nav (I)$. For an upper
bound on this,
we consider the best permutation, $I_B$, for \maj of an arbitrary sequence, $I$.
For $I_B$, \maj would buffer the half of the requests in the sequence with 
the highest frequencies.
The difference is
\begin{eqnarray}%
&& \maj(I_B) - \nav(I_B) \nonumber\\
&=& 2  \sum_{i=\lceil \frac{n+2}{2} \rceil}^{n}   f_I(a'_i) + \left(\left\lceil \frac{n}{2} \right\rceil - \left\lfloor \frac{n}{2} \right\rfloor \right) f_I(a'_{\lceil \frac{n}{2} \rceil}) - \sum_{i=1}^n f_I(a'_i)  \nonumber\\
 &=& \sum_{i=\lceil \frac{n+2}{2}\rceil}^n f_I(a'_i) - \sum_{i=1}^{\lfloor\frac{n}{2}\rfloor}f_I(a'_i). \nonumber %
\end{eqnarray}
This expression is exactly the same as the expression for $\nav (I) - \maj_W(I)$ from Eq.~\ref{eq:dif_navmaj}, so we get the same upper bound of $\frac{1}{4}$.
Now, for a lower bound on $\Max(\maj,\nav)$, we use the family of sequences, $I_n$ defined as
$$I_n = a_0,a_0,\ldots,a_0,a_1,a_2,\ldots,a_{\lfloor\frac{n}{2}\rfloor},$$ where there
are $\lceil \frac{n}{2} \rceil$ copies of $a_0$. Then
\[
\nav(I_n) = \left\lfloor\frac{n}{2}\right\rfloor\frac{1}{n} + \left\lceil \frac{n}{2} \right\rceil \frac{\lceil\frac{n}{2}\rceil}{n}
= \left\{ \begin{array}{ll}
\frac{n}{4} + \frac{1}{2} & \mbox{for even $n$}\\
\frac{n}{4} + 1 + \frac{1}{4n} & \mbox{for odd $n$}
\end{array} \right.
\]
and
\[
\maj(I_n) = n\frac{\lceil\frac{n}{2}\rceil}{n}=\left\lceil\frac{n}{2}\right\rceil.\]
$I_n$ gives a lower
bound of $\frac{1}{4}$ on $\Max(\maj,\nav)$, since
$\maj(I_n)-\nav(I_n) \geq \left\lceil\frac{n}{2}\right\rceil-
\frac{n}{4} - 1 - \frac{1}{4n}$.
It follows that, $\Min(\nav, \maj)= -\Max(\maj, \nav) = -\frac{1}{4}$,
and $l(\nav, \maj) = [-\frac{1}{4}, \frac{1}{4}]$.
\end{proof}

\subsection{Majority vs.\ Eager}

According to relative interval analysis, \maj has better performance
than \eag.

\begin{theorem}\label{thm:relint_maj_eag}
 According to relative interval analysis $l(\maj, \eag) = [-\frac{1}{2}, 1]$.
\end{theorem}
\begin{proof}
By Proposition~\ref{trivialMax}, $\Max(\maj,\eag)\leq 1$.
For the lower bound on $\Max(\maj,\eag)$, we consider the family of
sequences, $E_n$, from Definition~\ref{sequences}.
By Proposition~\ref{prop_sequences},
$\maj(E_n) - \eag(E_n) = (n-6+\frac{16}{n}) -2
= n - 8 + \frac{16}{n}$, and
$\Max(\maj,\eag)\geq \limsup_{n\to \infty} \frac{n-8+\frac{16}{n}}{n} = 1$.
Thus,
$\Max(\maj, \eag) = 1$.

For $\Min(\maj, \eag)$, we consider $\Max(\eag, \maj)$.
First we calculate an upper bound on $\eag(I) - \maj(I)$.
Suppose the input sequence $I$ of length $n$ gives the maximum value of $\eag(I) - \maj(I)$ over all sequences of length $n$.
Suppose $I$ has $k$ distinct items $a_1, a_2, a_3, \ldots , a_k$, and let $f_i=f_I(a_i)$ and $n_i = n_I(a_i)$ for all $i$. Assume that $f_1 \leq f_2 \leq f_3 \leq \ldots \leq f_k$, so $a_k$ is the most frequent item.
First, assume $n_k \leq \lceil \frac{n}{2} \rceil$.
\begin{eqnarray}
\eag(I) - \maj(I) &\leq& nf_k -  1 
\leq  n\frac{\lceil \frac{n}{2} \rceil}{n} - 1
\leq \frac{n}{2} - \frac{1}{2}
\label{eq:majcond}
\end{eqnarray}

It remains to consider the range $\lceil \frac{n}{2} \rceil < n_k \leq n$.
Assume for some positive integer $q$ that  $n_k = \lceil \frac{n}{2} \rceil + q$.
From Lemma~\ref{lem:wmaj}, we know that \maj's result has the lower bound $\maj_W(I) \geq 2(\sum_{i=1}^{k-1} n_i f_i + qf_k) + (\lceil \frac{n}{2} \rceil - \lfloor
\frac{n}{2} \rfloor)f_k$. The summation is minimized when the smallest $k-1$
frequencies are all equal to $\frac{1}{n}$. Since $k-1= \lfloor\frac{n}{2}
\rfloor-q$ in this case,
$\maj(I)\geq 2\left((\lfloor\frac{n}{2}\rfloor -q)\frac{1}{n} + q\frac{\lceil \frac{n}{2}\rceil+q}{n}\right)$. Hence,
\begin{eqnarray}
\eag(I) - \maj(I) &\leq& \left\lceil \frac{n}{2} \right\rceil + q -  2\left( \frac{1}{n}\left(\left\lfloor \frac{n}{2} \right\rfloor -q \right) + q \frac{\lceil \frac{n}{2} \rceil + q}{n}  \right)  \nonumber
\\
&=& \left\{ \begin{array}{ll}
\frac{n}{2} - 1 - \frac{2}{n}(q^2 - q) & \mbox{for even $n$}\vspace{1 mm} \\
\frac{n}{2} - \frac{1}{2} - \frac{1}{n}(2q^2 - q-1) & \mbox{for odd $n$}
\end{array} \right. \nonumber
\\
&\leq&
\frac{n}{2} - \frac{1}{2}
\label{eq:majcond2}
\end{eqnarray}
Thus, the same upper bound holds both when $n_k\leq \lceil\frac{n}{2}\rceil$ and
when $n_k > \lceil\frac{n}{2}\rceil$.

For a lower bound on the maximum value of $\eag(I) - \maj(I)$ for even $n$, we
use the input sequence $I =  a, a, a_1, a_2, a_3, a, a_4,a \ldots, a_{\frac{n}{2}}, a$ (an $a$ every second time after start-up).
For this sequence
\begin{eqnarray}
\eag(I) - \maj(I) &=& n\frac{1}{2} - \left( 4\frac{1}{2} + (n-4)\frac{1}{n}\right) \nonumber
\\
&=& \frac{n}{2} -3 + \frac{4}{n}.   \nonumber
\end{eqnarray}
For odd $n$, we add one $a$ at the end of the even length $I$ which gives $\eag(I) - \maj(I) = \frac{n}{2} -3 + \frac{5}{2n}$.
These lower bounds and the upper bounds from Eq.~\ref{eq:majcond}~and~\ref{eq:majcond2} are asymptotically all equal to $\frac{n}{2}$, so
\begin{equation}
\Min(\maj, \eag) = -\Max(\eag, \maj) = -\limsup _{n\to \infty} \frac{\eag (I) - \maj (I)}{n} = -\frac{1}{2}.\nonumber
\end{equation}
Therefore $l(\maj, \eag) = [-\frac{1}{2}, 1]$.

\end{proof}

\section{Relative Worst Order Analysis}\label{sec:rel_worst}

Relative worst order analysis~\cite{Boyar07} compares two online algorithms directly.
It compares two algorithms on their worst orderings of sequences which have the same content, but possibly different order.
The definition of this measure is somewhat more involved; see~\cite{BFL07j}
for more intuition on the various elements.
As in the case of competitive analysis, here too the relative performance of the algorithms depend on the length of the input sequence $I$.
As in Section~\ref{sec:comp}, we define a modified and more general version of relative worst order analysis. The definition is given for a maximization
problem, but trivially adaptable to be used for minimization problems as well;
only the decision as to when which algorithm is better would change.

The following definition is parameterized by a total ordering,
$\TO$, since we will later use it for both $\leq$ and $\geq$.
\begin{definition}\label{def:wor}
$f$ is a \emph{$\TRIP{\algA}{\algB}{\TO}$-function} if
\[\forall I\WEHAVE \algA_W(I) \TO (f(n)+o(f(n)))\cdot\algB_W(I),\]
where \algA and \algB are algorithms
and $\TO$ is a total ordering.
Recall from Definition~\ref{def:worstpermut} that the notation 
$\textsc{Alg}_W(I)$, where \textsc{Alg} is some algorithm,
denotes the result of \textsc{Alg} on its worst permutation of $I$.

$f$ is a \emph{bounding function with respect to} $\TRIP{\algA}{\algB}{\TO}$ if
$f$ is a $\TRIP{\algA}{\algB}{\TO}$-function and for any
$\TRIP{\algA}{\algB}{\TO}$-function $g$,
$\lim_{n\rightarrow\infty}\frac{f(n)}{g(n)}\TO 1$.

If $f$ is a bounding function with respect to $\TRIP{\algA}{\algB}{\leq}$
and $g$ is a bounding function with respect to $\TRIP{\algA}{\algB}{\geq}$, then
\algA and \algB are said to be \emph{comparable} if
$\lim_{n\rightarrow\infty}f(n)\leq 1$ or $\lim_{n\rightarrow\infty}g(n)\geq 1$.

If $\lim_{n\rightarrow\infty}f(n)\leq 1$, then $\algB$ is better than
$\algA$ and $g(n)$ is a
\emph{relative worst order function of $\algA$ and $\algB$}, and
if $\lim_{n\rightarrow\infty}g(n)\geq 1$, then $\algA$ is better than
$\algB$ and $f(n)$ is a
\emph{relative worst order function of $\algA$ and $\algB$}.
\end{definition}
We use $\WR_{\algA,\algB} = f(n)$ to indicate that $f(n)$ belongs to the
equivalence class of \emph{relative worst order functions of $\algA$ and 
$\algB$}.

The competitive function could also have been defined using this framework,
but was defined separately as a gentle introduction to
the idea.

\subsection{Naive vs.\ Optimal}\label{subsec:relwor_nav_opt}
Relative worst order analysis can show the strength of the simple, but
adaptive, \nav algorithm by comparing it with the powerful \opt.
\nav is an optimal algorithm according to relative worst order
analysis, in the sense that it is equivalent to \opt.
\begin{theorem}\label{thm:relwor_nav_opt}
According to relative worst order analysis $\WR_{\opt,\nav}=1$, so \nav and \opt are equivalent.
\end{theorem}
\begin{proof}
In the aggregate frequency problem, even though \opt knows the whole sequence in advance, it cannot store an item before it first appears in
the sequence.
Thus, for any input sequence $I$, the worst permutation for \opt is the sorting of $I$ according to the increasing order of the frequencies of the items, i.e., $D(I)$.
On this ordering, \opt is forced to behave like \nav.
Therefore, the constant function $1$ is a bounding function with
respect to both $(\opt,\nav,\leq)$ and $(\opt,\nav,\geq)$, so
$\WR_{\opt,\nav}=1$.
\end{proof}

\subsection{Naive vs.\ Eager}

According to relative worst order analysis, \nav is better than \eag.

\begin{theorem}\label{thm:relwor_nav_eag}
 According to relative worst order analysis $\WR_{{\nav},{\eag}} = \frac{n}{2}$.
\end{theorem}
\begin{proof}
From Theorem~\ref{thm:relwor_nav_opt}, we know that for \opt's worst 
permutation, $I_W$, of any sequence $I$, $\opt(I_W) = \nav(I_W)$.
Any arbitrary online algorithm $\algA$
cannot be better than \opt on any sequence, so \nav and \algA are 
comparable.
For any arbitrary online
algorithm $\algA$ and a worst order, $I_W$, for \algA of any sequence $I$, 
$\frac{\nav(I_W)}{\algA(I_W)} = \frac{\opt(I_W)}{\algA(I_W)}$, so
a competitive function for \algA is an upper bound on the relative worst order
function of \algA and \algB.
By Theorem~\ref{thm:comp_eag}, $\WR(\nav,\eag)\leq \frac{n}{2}$.
Consider the family of sequences, $E_n$, from Definition~\ref{sequences}. These sequences are in the worst ordering for both \eag and \opt.
By Proposition~\ref{prop_sequences}, $\nav(E_n) = n-4+\frac{8}{n}$ and $\eag(E_n)=2$.
Thus, $\nav(E_n) = \frac{n}{2}\eag(E_n) -4 + \frac{8}{n}$.
Consequently,
$\frac{n}{2}$ is a relative worst order function of \nav and \eag,
and $\WR_{{\nav},{\eag}} = \frac{n}{2}$.
\end{proof}

\subsection{Naive vs.\ Majority}

According to relative worst order analysis, \nav is better than
\maj, though not quite as much better as compared to \eag.

\begin{theorem}\label{thm:relwor_nav_maj}
According to relative worst order analysis, $\WR_{{\nav},{\maj}} = \frac{n}{4}$.
\end{theorem}
\begin{proof}
As in the proof of the previous theorem,
since \nav and \opt perform the same on their worst orderings
of any sequence, \nav and \maj are comparable.
Next we  derive a bounding function with respect to $(\nav,\maj,\leq)$.
Since \nav's aggregate frequency is the same on any ordering of that
sequence, we can compare \nav and \maj on the same sequence, \maj's
worst ordering of it; that is also a worst ordering for \nav.
Suppose the input sequence $I$ of length $n$ gives the largest ratio
for $\frac{\nav_W(I)}{\maj_W(I)}$ for sequences of length $n$.
Suppose $I$ has $k$ distinct items $a_1, a_2,\ldots, a_k$, and let $f_i=f_I(a_i)$ and $n_i = n_I(a_i)$
for all $i$.
Assume that $f_1 \leq f_2 \leq f_3 \leq \ldots \leq f_k$, so $a_k$ is the most frequent item.

If $n_k\leq \lfloor \frac{n}{2} \rfloor$ then
\begin{eqnarray}
\frac{\nav_W(I)}{\maj_W(I)} &=& \frac{\sum_{i=1}^k n_if_i}{2(\sum_{i=1}^{j-1}n_if_i + pf_j) + \left (\lceil \frac{n}{2} \rceil - \lfloor \frac{n}{2} \rfloor \right) f_j} \nonumber
\\
&=& \frac{\sum_{i=1}^k n_i^2}{2(\sum_{i=1}^{j-1}n_i^2 + pn_j) + \left (\lceil \frac{n}{2} \rceil - \lfloor \frac{n}{2} \rfloor \right) n_j} \label{eq:majw_profit}
\end{eqnarray}
where $j\leq k$ is the largest index such that $\sum_{i=1}^{j-1}n_i + p = \lfloor \frac{n}{2} \rfloor$ for some non-negative integer $p$.
Create another sequence $I'$ from $I$ by replacing all the $a_i$'s where $j<i<k$ with $a_k$ and by replacing $n_j - p - \left (\lceil \frac{n}{2} \rceil - \lfloor \frac{n}{2} \rfloor \right)$ $a_j$'s with $a_k$.
$I'$ will have $j+1$ distinct items and the most frequent item will have $\lfloor \frac{n}{2} \rfloor$ occurrences.
Since all these changes will increase the numerator and not change the denominator in Eq.~\ref{eq:majw_profit}, $I'$ will give at least as large a ratio as $I$, so we consider the sequence $I'$ instead of $I$.
Suppose the items of $I'$, in nondecreasing order of frequency, are $\hat{a}_1, \hat{a}_2,\ldots ,\hat{a}_{j+1}$ and the corresponding counts are $\hat{n}_1, \hat{n}_2,\ldots ,\hat{n}_{j+1}$. Then,
\begin{equation}
\frac{\nav_W(I')}{\maj_W(I')} \leq \frac{\lfloor \frac{n}{2} \rfloor^2 + \sum_{i=1}^j \hat{n}_i^2}{2\sum_{i=1}^{j}\hat{n}_i^2  - \left (\lceil \frac{n}{2} \rceil - \lfloor \frac{n}{2} \rfloor  \right) \hat{n}_j }\label{eq:majw_profit2}
\end{equation}

Consider any item $\hat{a}_i$ where $i\leq j$. Suppose its count is $\hat{n}_i >1$.
Replace the $\hat{n}_i$ copies of $\hat{a}_i$ by $\hat{n}_i$  distinct items which are different from all the other items in $I'$.
In most cases, this replacement will decrease the numerator in Eq.~\ref{eq:majw_profit2} by $\hat{n}_i^2 - \hat{n}_i$ and will decrease the denominator by $2(\hat{n}_i^2 - \hat{n}_i)$.
The only exception is when $i=j$ and $n$ is odd, in which case the denominator will decrease by $2\hat{n}_i^2 - 3\hat{n}_i +1$.
However, in either case, the decrease in the denominator is as large as that in the numerator.
Since the lower bound on the ratio is $1$, this replacement will increase the ratio.
Hence the maximum ratio will be achieved if all the items, except the most frequent item, have frequency $\frac{1}{n}$, so  $I'$ has the same form as $W_n$. Using Proposition~\ref{prop_sequences},
\begin{equation}
\frac{\nav_W(I')}{\maj_W(I')} = \left\{ \begin{array}{ll}
\frac{n}{4} + \frac{1}{2} & \mbox{for even $n$}\vspace{1 mm}\\
\frac{n}{4} + \frac{3}{4n} & \mbox{for odd $n$}
\end{array} \right. \label{eq:majw_up}
\end{equation}

It remains to consider the range $\lceil \frac{n}{2} \rceil \leq n_k \leq n$.
In this case,
\begin{eqnarray}
\frac{\nav_W(I)}{\maj_W(I)} &=& \frac{\sum_{i=1}^k n_if_i}{2(\sum_{i=1}^{k-1}n_if_i + qf_k) + \left (\lceil \frac{n}{2} \rceil - \lfloor \frac{n}{2} \rfloor \right) f_k} \nonumber
\\
&=& \frac{n_k^2 + \sum_{i=1}^{k-1} n_i^2}{2qn_k + 2\sum_{i=1}^{k-1}n_i^2 + \left (\lceil \frac{n}{2} \rceil - \lfloor \frac{n}{2} \rfloor \right) n_k} \label{eq:majw_profit3}
\end{eqnarray}
where $\sum_{i=1}^{k-1}n_i + q = \lfloor \frac{n}{2} \rfloor$ for some non-negative integer $q$.
As in the case of $n_k \leq \lfloor \frac{n}{2} \rfloor$, all the multiple instances of items other than $a_k$ can be replaced by distinct items with frequency $\frac{1}{n}$ without decreasing the ratio.
Next, if $q>0$ and we replace one instance of $a_k$ with some an item with frequency $\frac{1}{n}$, i.e., decrease $q$ by one, then the numerator in Eq.~\ref{eq:majw_profit3} will be decreased by $n_k^2 - (n_k-1)^2 - 1 = 2(n_k-1)$ and the denominator will be decreased by $$2qn_k - 2(q-1)(n_k-1)-2 +  \left\lceil \frac{n}{2} \right\rceil - \left\lfloor \frac{n}{2} \right\rfloor = 2(n_k + q -2) + \left\lceil \frac{n}{2} \right\rceil - \left\lfloor \frac{n}{2} \right\rfloor$$
Since the lower bound of the ratio is $1$, this replacement will increase the ratio while decreasing value of $q$.
Thus, the largest ratio will achieved when $q=0$, and
\begin{eqnarray}
\frac{\nav_W(I)}{\maj_W(I)} &\leq& \frac{\lceil \frac{n}{2} \rceil^2 + \lfloor \frac{n}{2}\rfloor}{2\lfloor \frac{n}{2}\rfloor  + \lceil \frac{n}{2} \rceil - \lfloor \frac{n}{2} \rfloor} \nonumber \\
& = & \left\{ \begin{array}{ll}
\frac{n}{4} + \frac{1}{2} & \mbox{for even $n$}\\
\frac{n}{4} + 1 - \frac{1}{4n} & \mbox{for odd $n$}
\end{array} \right. \label{eq:majw_up2}
\end{eqnarray}
By
Eqns.~\ref{eq:majw_up}~and~\ref{eq:majw_up2},
$\frac{n}{4}$ is a $(\nav,\maj,\leq)$-function.

Since the proof of the upper bounds above
shows that $W_n$ gives the largest ratio among sequences of length~$n$,
we can use the same sequence for the lower bound, showing that
$\frac{n}{4}$ is a bounding function with respect to $(\nav,\maj,\leq)$,
so $\WR_{\nav,\maj} = \frac{n}{4}$.

\end{proof}

\subsection{Majority vs.\ Eager}

\begin{theorem}\label{thm:relwor_maj_eag}
According to relative worst order analysis, \maj and \eag are incomparable.
\end{theorem}
\begin{proof}
First, we show that \maj can be much better than \eag.
Consider the family of sequences, $E_n$, from Definition~\ref{sequences}. These sequences are in their worst orderings
for both \maj and \eag. By Proposition~\ref{prop_sequences}, $\eag(E_n) = 2$, so $$\maj_W(E_n) = n-6+\frac{16}{n} \geq \left(\frac{n}{2}-3+\frac{8}{n}\right)\eag_W(E_n).$$

Now, we show that \eag can be much better than \maj.
Consider the family of sequences, $W_n$, from Definition~\ref{sequences}. These sequences are in their worst orderings
for \maj, so by Proposition~\ref{prop_sequences}, $\maj_W(W_n) = 1$. A worst ordering for \eag is
$$W'_n = a_1,a_2,\ldots,a_{\lceil\frac{n}{2}\rceil},a_0,a_0,\ldots,a_0,$$ where there
are $\lfloor \frac{n}{2} \rfloor$ copies of $a_0$. $\eag(W'_n) = \nav(W_n)$, which by Proposition~\ref{prop_sequences}
 is $\frac{n}{4} + \frac{1}{2}$ when $n$ is even and $\frac{n}{4} + \frac{3}{4n}$ when $n$ is odd.
Thus, $$\eag_W(W_n) \geq \frac{n}{4}\maj_W(W_n).$$

These two families of sequences show that \maj~ and \eag~ are incomparable under relative worst order analysis.

\end{proof}

\section{Conclusion and Future Work}
The frequent items problem for streaming was considered as an online
problem. Three deterministic algorithms, \nav, \maj, and \eag were
compared using three different quality measures: competitive analysis,
relative worst order analysis, and relative worst order ratio. According
to competitive analysis, \nav is the better algorithm and \maj and
\eag are equivalent. According to relative interval analysis, \nav
and \maj are equally good and both are better than \eag. According
to relative worst order analysis, \nav and \opt are equally good
and better than \maj and \eag, which are incomparable.

All three analysis techniques studied here are worst case measures.
According to both competitive analysis and relative worst order analysis,
\nav is the best possible online algorithm, and according the
relative worst order analysis, it is as good as \maj and better than
\eag. This is a consequence of \nav being very adaptive and, as a result,
good at avoiding the extreme poor performance cases.
Both \maj and \eag attempt to
keep the most frequent items in the buffer for longer
than their frequency would warrant. The heuristic approaches hurt
these algorithms in the worst case.

Relative interval analysis compares the algorithms on the same sequence
in a manner which, in addition to the worst case scenarios, also
takes the algorithms' best performance into account
to some extent. This makes \maj's sometimes superior performance visible,
whereas \eag, not being adaptive at all, does not benefit in the same way from
its best performance.
In some sense, \maj's behavior can be seen as swinging around the
behavior of \nav, with worse behavior on some sequences counter-acted
by correspondingly better behavior on other sequences.

Our conclusion is that purely worst behavior measures do not give
indicative results for this problem. Relative interval analysis does
better, and should possibly be supplemented by some expected case
analysis variant.
To that end, natural performance measures to consider would be
bijective and average analysis~\cite{Angelopoulos07}.
However, as the problem is stated in~\cite{Giannakopoulos12} and studied here,
the frequent items problem has an infinite universe from which the
items are drawn. Thus, these analysis techniques
cannot be applied directly to the problem in any meaningful way.
Depending on applications, it could be realistic to assume a finite universe.
This might give different results than those obtained here, and might allow
the problem to be studied using other measures, %
giving results dependent on the size of the universe.
Another natural extension of this work is to consider multiple buffers, which
also allows for a richer collection of algorithms~\cite{Berinde09},
or more complicated, not necessarily discrete, objective
functions~\cite{Cohen06}.

\bibliography{ref_stream}
\bibliographystyle{plain}

\end{document}